\documentclass[prl,twocolumn,showpacs,preprintnumbers,superscriptaddress]{revtex4}

\usepackage{graphicx}
\usepackage{color}
\usepackage{amsmath}
\usepackage{textcomp}
\usepackage{epstopdf}
\usepackage{revsymb}
\usepackage{amsbsy}
\usepackage{mathrsfs}
\usepackage{amsfonts}
\usepackage{amsthm}
\usepackage{float}
\usepackage{natbib}
\usepackage{xcolor}
\usepackage[colorlinks=true, pdfborder={0 0 0}, linkcolor=red,citecolor=blue, anchorcolor=blue]{hyperref}

\newcommand{\braket}[2]{\langle #1 | #2 \rangle}

\theoremstyle{plain}
\newtheorem{theorem}{Theorem}
\newtheorem{lemma}[theorem]{Lemma}

\theoremstyle{definition}

\begin{document}

\title{Broadcasting Quantum  Coherence via Cloning}
\author{Udit Kamal Sharma, Indranil Chakrabarty}
\affiliation{Center for Security$,$ Theory and Algorithmic Research$,$ International Institute of Information Technology$,$ Gachibowli$,$ Hyderabad$,$ India}
\author{Manish Kumar Shukla}
\affiliation{Center for Computational Natural Sciences and Bioinformatics$,$ International Institute of Information Technology$,$ Gachibowli$,$ Hyderabad$,$ India}

\begin{abstract}
\noindent{Quantum coherence has recently emerged as a key candidate for use as a resource in various quantum information processing tasks. Therefore, it is of utmost importance to explore the possibility of creating a greater number of coherent states from an existing coherent pair. In other words we  start with initial incoherent pair and induce coherence via quantum cloning. More specifically, we start with a genuinely incoherent state which remains incoherent with the change of basis and make it a coherent state at the end. This process is known as broadcasting of coherence via cloning, which can either be optimal or non-optimal. Interestingly, in this work for the first time we are able to give a method by which we can introduce coherence in the genuinely incoherent state. We use the computational basis representation of the most general two-qubit mixed state, shared between Alice and Bob, as the input state for the universal symmetric optimal Buzek-Hillery cloner. First of all we show that it is impossible to ensure optimal broadcast of coherence. Secondly, in case of non-optimal broadcasting, we show that the coherence introduced in the output states of the cloner will always be lesser than the initial coherence of the input state. Finally, we take the examples of statistical mixture of most coherent state (MCS) \& most incoherent state (MIX) and  Bell-diagonal states (BDS) to obtain their respective ranges of non-optimal broadcasting in terms of their input state parameters.}
\end{abstract}

\maketitle

\section{I. Introduction}

\noindent{\textbf{Quantum Coherence: }From the perspective of Quantum Information Theory, it is always interesting to ask what can be classified as a useful resource for information processing tasks. In last few decades, we have witnessed that quantum correlations \cite{correlation1, correlation2} and entanglement \cite{entanglement}, in particular, have lent a helping hand in revolutionizing many information processing tasks. In the recent years, another well-known quantum phenomenon, known as quantum coherence, has been identified as a candidate to act as a resource in quantum computation and quantum information processing.}\\ 

\noindent{Recently, a rigorous framework to quantify coherence has been proposed by Baumgratz \textit{et al.}\cite{baumgratz}.Some of the functionals which satisfy Baumgratz's framework include the relative entropy of coherence, the $l_1$-norm of coherence, the Wigner-Yanase-Dyson skew information \cite{Girolami} and Robustness of coherence\cite{Napoli, rob}.} \\

\noindent{Moreover, just like entanglement, a significant amount of research has been carried out to propose an operational resource theory of coherence\cite{resource1,resource2, resource3,resource4, resource5, resource6}. However, due to the basis dependent nature of quantum coherence, some researchers felt the need to address the notion of \textit{genuine quantum coherence},thereby, suggesting classes of genuinely incoherent operations\cite{genuine}. Apart from the research on single party systems, prospects of multipartite coherence like frozen quantum coherence and monogamy have been explored in \cite{Yao,Chandrashekhar}.}\\

\noindent \textbf{Quantum Cloning and Broadcasting} The impossibility  to clone quantum states is regarded as one of the most fundamental restriction that comes from nature\cite{wootters}.Even though we cannot copy an unknown quantum state perfectly, but quantum mechanics never rules out the possibility of cloning it approximately. It allows probabilistic cloning as one can always clone an arbitrary quantum state perfectly with some non-zero probability of success. In 1996, Buzek \textit{et al.} went beyond the idea of perfect cloning and introduced the concept of approximate state independent cloning \cite{buzek1}. 
Later, in another work by Gisin and Massar, this cloner was shown to be optimal \cite{brub, gisin2}.\\

\noindent In a recent work, it has been shown that it is impossible to optimally broadcast quantum correlation across different laboratories \cite{sourav}. But if we consider non-optimal broadcasting, it is possible to broadcast quantum correlation between separated parties \cite{Kher}. Apart from cloning, there are different ways to broadcast quantum correlation. Quite recently, many authors showed by using sophisticated methods that correlations in a single bipartite state can be locally or uni-locally broadcast if and only if the states are classical (i.e. having classical correlation) or classical-quantum respectively \cite{barnum, piani, barnum-gen, luo, luo-li}. \\


\noindent{\textbf{Motivation for Broadcasting Coherence} Unlike quantum entanglement, quantum coherence is a basis-dependent quantity. As an example, let us take the case of a 2-qubit state, $\rho_{12} = |00\rangle\langle00|$. Clearly, $\rho_{12}$ is an incoherent state in the computational basis $\{|00\rangle,|01\rangle, |10\rangle,|11\rangle\}$. But, when we write the same state in the Bell basis$\{|\Phi^+\rangle,|\Phi^-\rangle, |\Psi^+\rangle,|\Psi^-\rangle\}$, then, ${\rho}_{12} = \frac{1}{2}(|\Phi^+\rangle\langle\Phi^+| + |\Phi^+\rangle\langle\Phi^-| + |\Phi^-\rangle\langle\Phi^+| + |\Phi^-\rangle\langle\Phi^-|)$ becomes a coherent state because in the Bell-basis representation, $\rho_{12}$ has off-diagonal elements. So, now, we can have a very relevant question : What is the need for broadcasting coherence via cloning if we can always make coherence appear or disappear in the given state by a mere change of basis ?}\\ 

\noindent{In order to answer this question, let us take the example of the genuinely incoherent 2-qubit state, $\rho = \frac{I}{4}$, where $I$ is the $4 \times 4$ Identity matrix in the computational basis. An interesting property of this state is that it will always have the same representation in any basis of our choice, i.e, $\rho = \frac{1}{4}(|x_{1}\rangle\langle x_{1}| + |x_{2}\rangle\langle x_{2}| + |x_{3}\rangle\langle x_{3}| + |x_{4}\rangle\langle x_{4}|)$ in an arbitrary basis $\{|x_{1}\rangle,|x_{2}\rangle, |x_{3}\rangle,|x_{4}\rangle\}$. So, in such states, we cannot introduce coherence by merely changing the basis.}\\

\noindent{Therefore, in this work, one of our aim is to introduce coherence in the genuinely incoherent state by the process of cloning and broadcasting which is otherwise impossible with the change of basis. In the non-local cloning machine (see Fig \ref{NL}), we have two inputs : the state to be cloned ($\rho_{12}$) and the blank state ($\rho_{34}$), where the copy is made. The outputs of the cloning machine ($\tilde{\rho_{12}}$ and $\tilde{\rho_{34}}$) are independent of the nature of the blank state ($\rho_{34}$). Hence, we propose that we use the most incoherent 2-qubit state as the blank state ($\rho_{34}$) so that in the post-cloning scenario, we would be able to introduce coherence in it. A similar argument can be extended for the local cloning situation (described by Fig \ref{L}) where we can use single qubit state $\rho = \frac{I}{2}$ as the blank state in the qubits 3 and 4.
After introducing coherence, we can further convert the coherence in $\tilde{\rho_{34}}$) to entanglement\cite{Uttam}. We can even transform the coherence in $\tilde{\rho_{34}}$) to discord\cite{Ma}}.\\

\noindent{In this paper, our primary objective is to study the broadcast of coherence as a correlation present in a pair of qubits to other pairs via local and non-local cloning, with the pairs being spread out across two labs. 
Here, we have used the $l_1$ norm as a measure of coherence. Given a state $\rho$, with its matrix elements as $\rho_{ij}$ the amount of coherence 
present in the state $\rho$ in the basis $\{|i\rangle\}$ is given by the quantity 
$C(\rho)=\sum_{i,j=1}^d |\langle i|\rho |j\rangle|$, where $i \neq j$. It is interesting to note that coherence is a basis dependent quantity as the representation of coherence will be different in different basis. Since $l_1$-norm is a function of the off-diagonal elements of the given density matrix representation, clearly, it is evident that the value of the coherence will be zero in the eigenbasis of the density matrix, where there are no off diagonal elements. It is intuitive that for all other basis the physical findings followed by implications will be similar. So, without any loss of generality here we have 
computed the coherence in the two-qubit computational basis $\{|00\rangle,|01\rangle\, |10\rangle\, |11\rangle\}$.}

\section{II. Optimal and Non optimal broadcasting of coherence}

\noindent{Let us begin with a situation where we have two distant parties : Alice and Bob and they share a two-qubit mixed state $\rho_{12}$
which can be canonically expressed as,
\begin{eqnarray}
&&\rho_{12}=\frac{1}{4}\left[I_{2\times 2}\otimes I_{2\times 2}+\sum_{i=1}^{3}(x_{i}\sigma_{i}\otimes I)+\sum_{i=1}^{3}\right.\nonumber\\
&&(y_{i}I\otimes\sigma_{i})+\left.\sum_{i,j=1}^{3}(t_{ij}\sigma_{i}\otimes\sigma_{j})\right] = \left\{\overrightarrow{X},\:\overrightarrow{Y},\: T_{3\times 3}\right\},\:\:\: \label{eq:mix}
\end{eqnarray}
\noindent where $\overrightarrow{X}=\left\{ x_{1},\: x_{2},\: x_{3}\right\}$ and $\overrightarrow{Y}=\left\{ y_{1},\: y_{2},\: y_{3}\right\}$ are Bloch vectors with $0\leq\left\Vert \overrightarrow{x}\right\Vert \leq1$ and $0\leq\left\Vert \overrightarrow{y}\right\Vert \leq1$. Here, $t_{ij}$'s ($i,\:j$ = $\{1,2,3\}$) 
are elements of the correlation matrix ($T=[t_{ij}]_{3 \times 3}$), whereas $\sigma_i=(\sigma_1,\sigma_2,\sigma_3)$ are 
the Pauli matrices and $I$ is the identity matrix. We consider the $2$-qubit state
$\rho_{12}$ shared between Alice and Bob as the input state to the two types of cloners : universal symmetric optimal Buzek Hillery non-local cloner ($U_{bh}^{nl}$ or simply $U_{nl}$) and local cloner ($U_{bh}^{l}$ or simply $U_{l}$) (see Appendix A).}\\

\noindent We separately investigate both \textit{Optimal} and \textit{Non Optimal Broadcasting}. In order to broadcast coherence between the desired pairs $(1,2)$ and $(3,4)$ (see Fig \ref{L} and Fig \ref{NL}),  we have to maximize the amount of coherence between the non-local output pairs $(1,2)$ and $(3,4)$, regardless of the local output states between $(1,3)$ and $(2,4)$. However, for optimal broadcasting, we should ideally have no coherence between the local pair of qubits $(1,3)$ and $(2,4)$ in order to have the maximum amount of coherence in non local pair of qubits .\\

\noindent \textbf{Definition: Optimal Broadcasting (Local)}\\ 

\noindent{\label{def:broad_local }As shown in Fig \ref{L}, we start with two parties Alice and Bob sharing a coherent input state $\rho_{12}$, with qubits $1$ and $2$ of $\rho_{12}$ on Alice's side and Bob's side respectively. Also, qubits $3$ and $4$ serve as blank states ($|\Sigma_{3}\rangle \langle \Sigma_{3}| = \frac{I}{2}$ and $|\Sigma_{4}\rangle \langle \Sigma_{4}| = \frac{I}{2}$) on Alice's and Bob's side respectively. A coherent input state $\rho_{12}$ is said to be broadcast  optimally after the application of local unitary part  $U_{l}^{1} \otimes U_{l}^{2}$ of the cloning operation, 
on the qubits $1$ and $2$ respectively, if for some values of the input state parameters, the non-local output states, given by : } 
\begin{eqnarray*}
\label{eq1}
\begin{aligned}
\tilde{\rho}_{12} & = Tr_{34}\left[\tilde{\rho}_{1234}\right] \\
 & = Tr_{34}\left[U_{l}^{1}\otimes U_{l}^{2}\left(\rho_{12} \otimes |\Sigma_{3}\rangle \langle \Sigma_{3}| \otimes |\Sigma_{4}\rangle \langle \Sigma_{4}| \right)\right],
\end{aligned}\\
\begin{aligned}
\tilde{\rho}_{34} & = Tr_{12}\left[\tilde{\rho}_{1234}\right] \\
 & = Tr_{12}\left[U_{l}^{1}\otimes U_{l}^{2}\left(\rho_{12} \otimes |\Sigma_{3}\rangle \langle \Sigma_{3}| \otimes |\Sigma_{4}\rangle \langle \Sigma_{4}| \right)\right],
\end{aligned}
\end{eqnarray*} 

\noindent{are coherent i.e $C(\tilde{\rho}_{12}) \neq 0, C(\tilde{\rho}_{34}) \neq 0$, while the  local output states, given by :}
\begin{eqnarray*}
\label{eq2}
\begin{aligned}
\tilde{\rho}_{13} & = Tr_{24}\left[\tilde{\rho}_{1234}\right] \\
 & = Tr_{24}\left[U_{l}^{1}\otimes U_{l}^{2}\left(\rho_{12} \otimes |\Sigma_{3}\rangle \langle \Sigma_{3}| \otimes |\Sigma_{4}\rangle \langle \Sigma_{4}| \right)\right],
\end{aligned}\\
\begin{aligned}
\tilde{\rho}_{24} & = Tr_{13}\left[\tilde{\rho}_{1234}\right] \\
 & = Tr_{13}\left[U_{l}^{1}\otimes U_{l}^{2}\left(\rho_{12} \otimes |\Sigma_{3}\rangle \langle \Sigma_{3}| \otimes |\Sigma_{4}\rangle \langle \Sigma_{4}| \right)\right],
\end{aligned}
\end{eqnarray*} 

\noindent{are incoherent i.e $C(\tilde{\rho}_{13}) = 0, C(\tilde{\rho}_{24}) = 0$.}\\

\begin{figure}[htbp]
\begin{center}
\includegraphics[height=8cm,width=7cm]{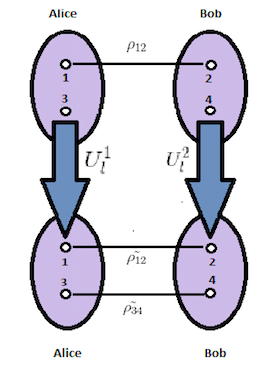}
\end{center}
\caption{\noindent \scriptsize
Application of local cloning operations $U_{l}^1$ and $U_{l}^2$ on the qubits 1 and 2 of the input state $\rho_{12}$ and the blank states $|\Sigma_{3}\rangle \langle \Sigma_{3}|$ and $|\Sigma_{4}\rangle \langle \Sigma_{4}|$ on Alice's and Bob's labs respectively to get the output states $\tilde{\rho}_{12}$ and $\tilde{\rho}_{34}$.
}
\label{L}
\end{figure}

\noindent{\textbf{Definition: Optimal Broadcasting (Non-local)}}\\ 

\noindent{\label{def:broad_nonlocal}
As shown in Fig \ref{NL}, we start with two parties Alice and Bob sharing a coherent input state $\rho_{12}$, with qubits $1$ and $2$ on Alice's side and Bob's side respectively; and a blank state $\rho_{34} = \frac{I}{4}$, with qubits $3$ and $4$ jointly shared by both of them. 
A coherent input state $\rho_{12}$ is said can be broadcasted  optimally after the application of non-local unitary part $U_{nl}$ of the cloning operation together on qubits 1 and 2, if for some values of the input state parameters, the non-local output states, given by : }
\begin{eqnarray*}
&& \tilde{\rho}_{12}= Tr_{34}\left[\tilde{\rho}_{1234}\right]=Tr_{34}\left[U_{nl}\left(\rho_{12}\otimes \rho_{34}\right)\right], {}\nonumber\\&&
\tilde{\rho}_{34} = Tr_{12}\left[\tilde{\rho}_{1234}\right]=Tr_{12}\left[U_{nl} \left(\rho_{12}\otimes \rho_{34}\right) \right],
\end{eqnarray*}
\noindent{are coherent i.e $C(\tilde{\rho}_{12}) \neq 0$,and  $C(\tilde{\rho}_{34}) \neq 0$, while the  local output states, given by :}
\begin{eqnarray*}
&&\tilde{\rho}_{13} =Tr_{24}\left[\tilde{\rho}_{1234}\right]= Tr_{24}\left[U_{nl}\left(\rho_{12}\otimes \rho_{34}\right)\right],{}\nonumber\\&&
\tilde{\rho}_{24} = Tr_{13}\left[\tilde{\rho}_{1234}\right]= Tr_{13}\left[U_{nl}\left(\rho_{12}\otimes \rho_{34}\right)\right],
\end{eqnarray*}
\noindent{are incoherent i.e $C(\tilde{\rho}_{13}) = 0$ and  $C(\tilde{\rho}_{24}) = 0$.}\\

\begin{figure}[htbp]
\begin{center}
\includegraphics[height=8cm,width=6cm]{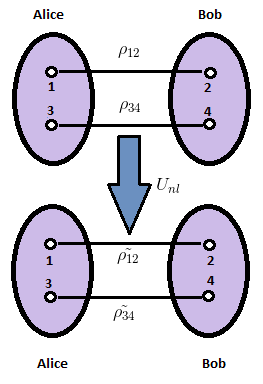}
\end{center}
\caption{\noindent \scriptsize
Application of non-local cloning operation $U_{nl}$ on the input state $\rho_{12}$ and blank state $\rho_{34}$ to get the output states $\tilde{\rho}_{12}$ and $\tilde{\rho}_{34}$.
}
\label{NL}
\end{figure}

\noindent{Interestingly, we find that it is impossible to broadcast coherence optimally via both non-local and local cloning transformations. We show this result for the computational basis $\{|00\rangle,|01\rangle\, |10\rangle\, |11\rangle\}$ with the general mixed state (given in Eq.~\eqref{eq:mix}) as the input resource without any loss of generality.}\\

\noindent \begin{theorem}
\label{theo:local_broad_no_optimal}
\textit{Given a two qubit general mixed state $\rho_{12}$ and B-H local cloning 
transformations 
(see Appendix A), 
it is impossible to optimally broadcast the coherence within $\rho_{12}$ into two coherent states: $\tilde{\rho}_{12}$ = $Tr_{34}\left[U_{l}^{1}\otimes U_{l}^{2}\left(\rho_{12}\otimes |\Sigma_{3}\rangle \langle \Sigma_{3}| \otimes |\Sigma_{4}\rangle \langle \Sigma_{4}|\right)\right]$, and
$\tilde{\rho}_{34}$ = $Tr_{12}\left[U_{l}^{1} \otimes U_{l}^{2}\left(\rho_{12}\otimes |\Sigma_{3}\rangle \langle \Sigma_{3}| \otimes |\Sigma_{4}\rangle \langle \Sigma_{4}|\right) \right]$, where $|\Sigma_{3}\rangle \langle \Sigma_{3}| =  |\Sigma_{4}\rangle \langle \Sigma_{4}| = \frac{I}{2}$.}
\end{theorem}

\begin{proof} 
\noindent Here, we consider the input state as the most general two qubit mixed state $\rho_{12}$. The  local unitary part  $U_{l}^{1}\otimes U_{l}^{2}$ of the B-H cloning transformation (see Appendix A) is applied on qubits of $\rho_{12} = \{ x, y, T \}$. We have local output states 
as, $\tilde{\rho}_{13} = \{\mu x, \mu x, T_{l} \} $ and $\tilde{\rho}_{24} = \{\mu y, \mu y, T_{l} \} $, where $T_{l} = diag(2\lambda, 2\lambda, 1-4\lambda)$ is its 3x3 diagonal correlation matrix and the non-local output states 
$\tilde{\rho}_{12}=\tilde{\rho}_{34}=\{\mu x, \mu y, \mu^2T\}$. Here, $\mu=1-2\lambda$  and $\lambda$ being the cloning machine parameter (See Appendix A) and $x,y$ represent the Bloch vectors and $T$ represents the 3x3 correlation matrix of the input state. The coherence calculated by using $l_1$ norm of the local output states in the computational basis is given by,
$C(\tilde{\rho}_{13})=2\sqrt{x_1^2+x_2^2}+2\lambda(1-2\sqrt{x_1^2+x_2^2})>0$ on Alice's side and $C(\tilde{\rho}_{24})=2\sqrt{y_1^2+y_2^2}+2\lambda(1-2\sqrt{y_1^2+y_2^2})>0$ on Bob's side. Each of these quantities are positive as the norm $||\vec{v}|| = \sqrt{v_1^2+v_2^2}$ for a 2d vector $\vec{v} = \{v_{1}, v_{2}\}$ is a non-negative quantity and in the minimum case of $x_{1} = x_{2} = 0$ and $y_{1} = y_{2} = 0$, still, $C(\tilde{\rho}_{13}) = C(\tilde{\rho}_{24}) = 2\lambda > 0$. This is true for 
all values of $\lambda$ lying between $0$ and $1/2$.\\

\noindent Now for $\lambda = \frac{1}{6}$, the above B-H local cloner 
reduces to state independent B-H optimal local cloner. Subsequently, the coherence of local output states changes to,
$C(\tilde{\rho}_{13})=1/3+(4/3)\sqrt{x_1^2+x_2^2})>0$ and $C(\tilde{\rho}_{24})=1/3+(4/3)\sqrt{y_1^2+y_2^2})>0$. Since the coherence of local output states is non-vanishing, it is clear enough that it is impossible to optimally broadcast quantum coherence via local cloning.
\end{proof}

\noindent \begin{theorem}
\label{theo:nonlocal_broad}
\textit{Given the two-qubit general mixed state $\rho_{12}$ and B-H non-local cloning 
transformations 
(see Appendix A), 
it is impossible to optimally broadcast the coherence within $\rho_{12}$ into two coherent states: $\tilde{\rho}_{12}$ = $Tr_{34}\left[U^{nl}_{bh}\left(\rho_{12}\otimes \rho_{34}\right)\right]$ and $\tilde{\rho}_{34}$ = $Tr_{12}\left[U^{nl}_{bh}\left(\rho_{12}\otimes \rho_{34}\right) \right]$, where $\rho_{34} = \frac{I}{4}$.}
\end{theorem}

\begin{proof} 
\noindent Here, we start with most general version of the two qubit mixed state $\rho_{12} = \{ x, y, T \}$. When  non-local unitary part $U_{bh}^{nl}$ or $U_{nl}$ (see Fig \ref{NL}) of the B-H cloning transformation is applied on  qubits of  $\rho_{12}$, then, we have local output states
as, $\tilde{\rho}_{13} = \{\mu x, \mu x, T_{nl}\}$ and $\tilde{\rho}_{24} = \{\mu y, \mu y, T_{nl}\}$, where $T_{nl} = diag(2\lambda, 2\lambda, 1-8\lambda)$ is its 3x3 diagonal correlation matrix and the non-local output states 
$\tilde{\rho}_{12}=\tilde{\rho}_{34}=\{\mu x, \mu y, \mu T\}$. Here, $\mu=1-4\lambda$ with $\lambda$ being the cloning machine parameter (see Appendix A) and $x,y$ represent the Bloch vectors and $T$ represents the correlation matrix of the input state. The coherence is calculated by using $l_1$ norm measure in the computational basis. Hence, for the local output states, we have,
$C(\tilde{\rho}_{13})=2\sqrt{x_1^2+x_2^2}+2\lambda(1-4\sqrt{x_1^2+x_2^2})>0$ on Alice's side and $C(\tilde{\rho}_{24})=2\sqrt{y_1^2+y_2^2}+2\lambda(1-4\sqrt{y_1^2+y_2^2})>0$ on Bob's side. Each of these quantities are positive as the norm $||\vec{v}|| = \sqrt{v_1^2+v_2^2}$ for a 2d vector $\vec{v} = \{v_{1}, v_{2}\}$ is a non-negative quantity and in the minimum case of $x_{1} = x_{2} = 0$ and $y_{1} = y_{2} = 0$, still, $C(\tilde{\rho}_{13}) = C(\tilde{\rho}_{24}) = 2\lambda > 0$. This is true for 
all values of $\lambda$ lying between $0$ and $1/4$.\\

\noindent Now, for $\lambda= \frac{1}{10}$, the above B-H non-local cloning machine 
reduces to state independent B-H optimal non-local cloner. Subsequently, the coherence of local output states changes to,
$C(\tilde{\rho}_{13})=1/5+(6/5)\sqrt{x_1^2+x_2^2})>0$ and $C(\tilde{\rho}_{24})=1/5+(6/5)\sqrt{y_1^2+y_2^2})>0$. Therefore, we see that it is impossible to broadcast quantum coherence optimally as we have non-vanishing quantum coherence for the local output states.\\
\end{proof}

\noindent At this point, we note that it is impossible to optimally  broadcast coherence between the desired qubits. However, we cannot rule out the entire possibility of broadcasting. If we relax the condition $C(\tilde{\rho}_{13}) = 0, C(\tilde{\rho}_{24}) = 0$, then we show that it will be possible to broadcast coherence non-optimally.\\

\noindent \textbf{Definition : Non-Optimal Broadcasting (Local)} 
A coherent state $\rho_{12}$ can be broadcasted non-optimally after the application of local unitary part  $U_{l}^{1} \otimes U_{l}^{2}$ (see Fig \ref{L})of the  cloning operation, if for some values of the input state parameters, the coherence of non-local output states is greater than 
coherence of local output states. Mathematically, for non-optimal broadcasting via local cloning, we should have the following conditions :
\begin{itemize}
    \item $C(\tilde{\rho}_{12}) > C(\tilde{\rho}_{13})$ and $C(\tilde{\rho}_{12}) > C(\tilde{\rho}_{24})$
    \item $C(\tilde{\rho}_{34}) > C(\tilde{\rho}_{13})$ and $C(\tilde{\rho}_{34}) > C(\tilde{\rho}_{24})$
\end{itemize}

\noindent \textbf{Definition : Non Optimal Broadcasting (Non-local)} 
A coherent state $\rho_{12}$ can be broadcasted non-optimally after the application of non-local unitary part $U_{nl}$ (see Fig \ref{NL}) of the cloning operation, if for some values of the input state parameters, the coherence of non-local output states is greater than 
coherence of local output states. Mathematically, for non-optimal broadcasting via non-local cloning, we should have the following conditions :
\begin{itemize}
    \item $C(\tilde{\rho}_{12}) > C(\tilde{\rho}_{13})$ and $C(\tilde{\rho}_{12}) > C(\tilde{\rho}_{24})$
    \item $C(\tilde{\rho}_{34}) > C(\tilde{\rho}_{13})$ and $C(\tilde{\rho}_{34}) > C(\tilde{\rho}_{24})$
\end{itemize}

\noindent Under this relaxation, the quantum coherence of the non-local output states $\tilde{\rho}_{12}$ and $\tilde{\rho}_{34}$ can either increase or decrease from the initial resource state $\rho_{12}$, which implies the following two scenarios after broadcasting :

\begin{itemize}
\item $C\left(\tilde{\rho}_{34}\right)$ $<$ $C\left(\rho_{12}\right)$ and $C\left(\tilde{\rho}_{12}\right)$ $<$ $C\left(\rho_{12}\right)$
\item $C\left(\tilde{\rho}_{34}\right)$ $>$ $C\left(\rho_{12}\right)$ and $C\left(\tilde{\rho}_{12}\right)$ $>$ $C\left(\rho_{12}\right)$
\end{itemize}

\noindent Interestingly, in this work we show that it is indeed impossible to increase the quantum coherence by both non-local and local cloning operation. We present this result in form of the following theorems by taking the two-qubit most general state (described in Eq \eqref{eq:mix}) as the input state. We have measured the coherence in the following proofs in computational basis without any loss of generality.

\noindent \begin{theorem}
\label{theo:local_broad}
\textit{Given a two qubit general mixed state $\rho_{12}$ and B-H local cloning transformation 
(see Appendix A), we can always create two 
states  $\tilde{\rho}_{12}$ and $\tilde{\rho}_{34}$, such that :
(a) $C(\tilde{\rho}_{12}) < C(\rho_{12})$
(b) $C(\tilde{\rho}_{34}) < C(\rho_{12})$ in the computational basis $\{|0\rangle,|1\rangle\}$, where $\mu = 1- 2 \lambda$, with $\lambda$ being the machine parameter.}
\end{theorem}

\begin{proof} 
\noindent We start with most general representation of two qubit mixed state $\rho_{12} = \{ x, y, T\}$ . When Buzek Hillary local cloning transformation $U_{l}^{1} \otimes U_{l}^{2}$ is applied to  $\rho_{12}$, we get the non-local output states as 
$\tilde{\rho}_{12} =\tilde{\rho}_{34}=\{\mu x, \mu y, \mu^2 T\}$. Here, $\mu=1-2\lambda$ with $\lambda$ being the cloning machine parameter and $x,y$ represent the Bloch vectors and $T=[t_{ij}]_{3 \times 3}$ represents the correlation matrix of the input state. Using $l_1$-norm as measure of coherence in the computational basis, we have,
$C(\rho_{12}) 
= 1/2[a_1 + a_2+ a_3]$
and 
$C(\tilde{\rho}_{12}) = C(\tilde{\rho}_{34})  = 1/2[b_1 + b_2 + b_3]$.

\noindent Here,
\begin{eqnarray}
a_1 &=& \sqrt{(t_{13}+t_{21})^2+(t_{11}-t_{22})^2}{}\nonumber\\&& +\sqrt{(t_{12}-t_{21})^2+(t_{11}+t_{22})^2}{}\nonumber\\
a_2 &=& \sqrt{(t_{13}-x_1)^2+(t_{23}-x_2)^2}{}\nonumber\\&&
+ \sqrt{(t_{13}+x_1)^2+(t_{23}+x_2)^2}{}\nonumber\\
a_3 &=& \sqrt{(t_{31}-y_1)^2+(t_{32}-y_2)^2}{}\nonumber\\ &&
+\sqrt{(t_{31}+y_1)^2+(t_{32}+y_2)^2}{}\nonumber\\
b_1 &=& \mu^2\sqrt{(t_{13}+t_{21})^2+(t_{11}-t_{22})^2}{}\nonumber\\ &&
+ \mu^2\sqrt{(t_{12}-t_{21})^2+(t_{11}+t_{22})^2}{}\nonumber\\
b_2 &=& \mu\sqrt{(\mu t_{13}-x_1)^2+(\mu t_{23}-x_2)^2}{}\nonumber\\&&
+\mu\sqrt{(\mu t_{13}+x_1)^2+(\mu t_{23}+x_2)^2}{}\nonumber\\
b_3 &=& \mu\sqrt{(\mu t_{31}-y_1)^2+(\mu t_{32}-y_2)^2}{}\nonumber\\&&
+\mu\sqrt{(\mu t_{31}+y_1)^2+(\mu t_{32}+y_2)^2}.{}\nonumber\\&&
\end{eqnarray}
\noindent If we can show that the conditions $a_1 > b_1$, $a_2 > b_2$ and $a_3 > b_3$ hold true at the same time, then we can always conclude that
$C({\rho}_{12}) $ $>$ $C(\tilde{\rho}_{12})$ and $C({\rho}_{12})  > C(\tilde{\rho}_{34})$.\\

\noindent  At first, if we observe the terms $a_1$ and $b_1$ carefully, we find that $b_1$ $=$ $\mu^2 a_1$ $<$ $a_1$ because $\mu=1-2\lambda < 1$ and $a_1 \geq 0$.\\

\noindent Now, we need to compare $a_2$ and $b_2$. If we observe the terms in $a_2$ and $b_2$, and ignoring the constant $1/2$, they can be thought of as distance between two points in a Cartesian Plane. Let us denote the distance between two points A and B in a 2-dimensional space (like XY Cartesian Plane) by $\|AB\|$. Let us consider
$\tilde{b_2} = \sqrt{(\mu t_{13}-x_1)^2+(\mu t_{23}-x_2)^2}{}+\sqrt{(\mu t_{13}+x_1)^2+(\mu t_{23}+x_2)^2}{}$ and as $\mu<1$, so, we have, $b_2 < \tilde{b_2}$. Now, assume 4 points in Cartesian Plane, namely, $A(t_{13},t_{23})$, $B(x_1,x_2)$,$C(-x_1,-x_2)$ and $D(\mu t_{13},\mu t_{23})$ without $O$ as the origin. Now, as $\mu<1$, hence, it is always true that $\|OD\| = \sqrt{(\mu t_{13})^2+(\mu t_{23})^2}{} = \mu\sqrt{(t_{13})^2+( t_{23})^2}{} = \mu \|OA\|$; that is, $\|OD\| < \|OA\|$. So, D always lies on $\overline{OA}$, inside $\Delta ABC$.
\begin{figure}[htbp]
\begin{center}
\[
\begin{array}{cc}
\includegraphics[height=7cm,width=6.5cm]{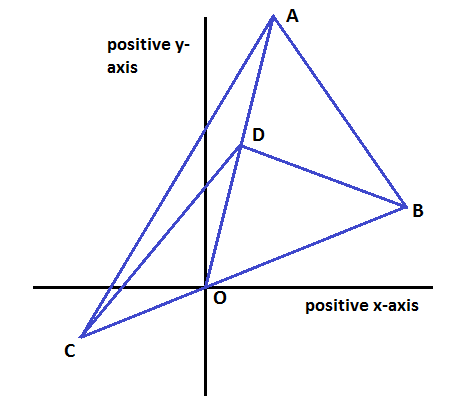} 
\end{array}
\]
\end{center}
\caption{\noindent \scriptsize
$\Delta ABC$  in XY-Plane, with point D in its interior}
\label{triangle}
\end{figure}
\noindent Now, from the Figure \ref{triangle}, we have, $\tilde{b_2}=\|DC\| + \|DB\|$ and $a_2 =\|AC\| + \|AB\|$. Now, we have proved $\tilde{b_2} < a_2$ in the Appendix B by plane geometry.\\

\noindent Hence,we have, $ b_2 < \tilde{b_2} < a_2$. Similarly, we can show $ b_3 < a_3$. Hence, we have proved that $ b_1 + b_2 + b_3 < a_1 + a_2 + a_3$.

\end{proof}

\noindent \begin{theorem}
\label{theo:non_local_broad}
\textit{Given a two qubit general mixed state $\rho_{12}$ and B-H non-local cloning transformation 
(see Appendix A), we can always create two states $\tilde{\rho}_{12}$ and $\tilde{\rho}_{34}$, such that 
:(a) $C(\tilde{\rho}_{12}) = \mu C(\rho_{12})< C(\rho_{12}) $
(b) $C(\tilde{\rho}_{34}) = \mu C(\rho_{12})<C(\rho_{12})$ in the computational basis, where $\mu = 1- 4 \lambda$, with $\lambda$ being the machine parameter.}
\end{theorem}
\begin{proof} First of all we apply  Buzek Hillary non-local cloning transformation 
the 2- qubit general mixed quantum states $\rho_{12} = \{ x, y, T \}$. The non-local output states are given by, 
$\tilde{\rho}_{12}=\tilde{\rho}_{34}=\{\mu x, \mu y, \mu T\}$. Here, $\mu=1-4\lambda$ with $\lambda$ being the cloning machine parameter, $x,y$ represent the Bloch vectors and 
$T=[t_{ij}]_{3 \times 3}$ represents the correlation matrix of the input state. The coherence present in these states, calculated in the computational 
basis by using the $l_1$ norm, are given by,
$C(\rho_{12}) = \frac{1}{2}X$
and 
$C(\tilde{\rho}_{12}) = C(\tilde{\rho}_{34}) = 
\frac{\mu}{2} X$,
where, 
\begin{eqnarray}
X&=&[\sqrt{(t_{13}+t_{21})^2+(t_{11}-t_{22})^2} {}\nonumber\\&&
+ \sqrt{(t_{12}-t_{21})^2+(t_{11}+t_{22})^2}{}\nonumber\\&&
+\sqrt{(t_{13}-x_1)^2+(t_{23}-x_2)^2}{}\nonumber\\&&
+\sqrt{(t_{13}+x_1)^2+(t_{23}+x_2)^2}{}\nonumber\\&&
+\sqrt{(t_{31}-y_1)^2+(t_{32}-y_2)^2}{}\nonumber\\&&
+\sqrt{(t_{31}+y_1)^2+(t_{32}+y_2)^2}]
\end{eqnarray}
\noindent Now, we know, $\mu = 1 - 4\lambda < 1$ and   $\lambda < 1/4$. So, clearly, $C(\tilde{\rho}_{12}) = \mu C(\rho_{12}) < C(\rho_{12})$ and  
$C(\tilde{\rho}_{34}) = \mu C(\rho_{12}) < C(\rho_{12})$. Hence, we have created two correlations across the labs of Alice and Bob with lesser coherence compared to that of the input state via non-local cloning.
\end{proof}

\section{III. Non-Optimal broadcasting for particular mixed states}

\noindent In this section, we consider the standard examples of two qubit mixed state as a resource state. These include states representing a mixture of the most coherent state (MCS) and the most incoherent state (MIS), and Bell diagonal states (BDS). We apply the standard BH cloner non-locally, to find the range where non-optimal broadcasting is possible, in terms of input state parameters. In all the examples, we have chosen computational basis for measuring coherence via $l_1$-norm.\\
 
\noindent \textbf{\textit{\centerline{Example: A. Mixture of MIS and MCS:}}}\\

\noindent Here, we consider a two-qubit state representing a convex mixture of the most coherent two-qubit state  ($\rho_{MCS}$) and the most incoherent two-qubit state ($\rho_{MIS}$) in computational basis. In computational basis, we have, $\rho_{MIS} = \frac{I}{4}$, where $I$ is a 4x4 identity matrix and $\rho_{MCS} = |MCS\rangle\langle MCS|$, where $|MCS\rangle = \frac{1}{2}(|00\rangle + |01\rangle + |10\rangle + |11\rangle)$. Now, assuming $p$ as the mixing parameter, we have $\rho_{12} = p \rho_{MCS} + (1 - p) \rho_{MIS}$, where $0 \leq p \leq 1$. Formally, we represent this class of states in terms of bloch vectors as $\rho_{12} = \{ \vec{x}^{MIX},\vec{y}^{MIX}, T^{MIX} \}$, where $\vec{x}^{MIX} = \{p, 0, 0\}$ and $\vec{y}^{MIX} = \{p, 0, 0\}$ are the Bloch vectors, and $T^{MIX}= \left(\begin{smallmatrix}
p&0&0\\ 0&0&0\\0&0&0
\end{smallmatrix} \right)$ is 3x3 correlation matrix with $p$ being the same mixing parameter.\\

\textit{ \centerline{Local Cloning} }\\

\noindent Here, we use the general local cloning machine 
(see Appendix A) to obtain the local output states as, $\tilde{\rho}_{13} = \big\{\mu\vec{x^{MIX}},\mu\vec{x^{MIX}},T_{l}\big\}$ and $\tilde{\rho}_{24} = \big\{\mu\vec{y^{MIX}},\mu\vec{y^{MIX}},T_{l}\big\}$  and the non-local output states as,
$\tilde{\rho}_{12} = \tilde{\rho}_{34} = \big\{\mu\vec{x^{MIX}},\mu\vec{x^{MIX}},\mu^2 T^{MIX}\big\}$.
\noindent Now, we have $\mu = 1-2\lambda$, with $\lambda$ being the cloning machine parameter and $T^{MIX}$ is same as the correlation matrix of $\rho_{12}$. Here, $x^{MIX}$ and $y^{MIX}$ are the input state bloch vectors for the $\rho_{12}$ as described above and $T_{nl} = diag[2\lambda, 2\lambda, 1-4\lambda]$ is the 3x3 diagonal correlation matrix of the local output states.\\

\noindent Using $l_1$-norm, the coherence for the local output states is given by, $C(\tilde{\rho}_{13}) = C(\tilde{\rho}_{24}) = 2\lambda + 2p(1 - 2\lambda)$; while, that of the non-local output states is given by,
$C(\tilde{\rho}_{12}) = C(\tilde{\rho}_{34}) = p(1 - 2 \lambda)(3 - 2 \lambda)$.
The expression for the condition of broadcasting in the general local cloning case reduces to
$p(1 - 2 \lambda)^{2} > 2\lambda$.\\

\noindent To obtain the state independent version, we substitute the value of $\lambda$ as $\frac{1}{6}$ (see Appendix A). In this case, for the local output states. we have $C(\tilde{\rho}_{13}) = C(\tilde{\rho}_{24}) = \frac{1 + 4p}{3}$; and for the non-local output states, we have $C(\tilde{\rho}_{12}) = C(\tilde{\rho}_{34}) = \frac{16 p}{9}$.
\noindent Hence, the  broadcasting in state independent local cloning case is possible when $p > \frac{3}{4}$.\\

\textit{ \centerline{Non-local Cloning} }\\

\noindent Here, we use the general non-local cloning machine 
(see Appendix A) to obtain the local output states as, $\tilde{\rho}_{13} = \big\{\mu\vec{x^{MIX}},\mu\vec{x^{MIX}},T_{nl}\big\}$ and $\tilde{\rho}_{24} = \big\{\mu\vec{y^{MIX}},\mu\vec{y^{MIX}},T_{nl}\big\}$ and the non-local output states as,
$\tilde{\rho}_{12} = \tilde{\rho}_{34} = \big\{\mu\vec{x^{MIX}},\mu\vec{x^{MIX}},\mu T^{MIX}\big\}$ .
\noindent Now, we have $\mu = 1-4\lambda$, with $\lambda$ being the cloning machine parameter and $T^{MIX}$ is same as the correlation matrix of $\rho_{12}$. Here, $x^{MIX}$ and $y^{MIX}$ are the input state bloch vectors for the $\rho_{12}$ as described above and $T_{nl} = diag[2\lambda, 2\lambda, 1-8\lambda]$ is the 3x3 diagonal correlation matrix of the local output states.\\

\noindent Using $l_1$-norm, the coherence for the local output states is given by, $C(\tilde{\rho}_{13}) = C(\tilde{\rho}_{24}) = 2\lambda + 2p(1 - 4\lambda)$; while, that of the non-local output states is given by,
$C(\tilde{\rho}_{12}) = C(\tilde{\rho}_{34}) = 3p(1 - 4\lambda)$.
The broadcasting in the general non-local cloning case is possible when 
$p > 2\lambda + 4p\lambda$.\\

\noindent To obtain the state independent version, we substitute the value of $\lambda$ as $\frac{1}{10}$ (see Appendix A). In this case, for the local output states. we have $C(\tilde{\rho}_{13}) = C(\tilde{\rho}_{24}) = \frac{1 + 6p}{5}$; and for the non-local output states, we have $C(\tilde{\rho}_{12}) = C(\tilde{\rho}_{34}) = \frac{9p}{5}$.
\noindent Hence, the condition for broadcasting in state independent non-local cloning case reduces to $p > \frac{1}{3}$.\\

\noindent Hence, in this example , it is evident that with respect to computational basis, non-local state independent cloning machine performs much better than state independent local cloning machine when it comes to broadcasting of coherence.\\

\noindent \textbf{\textit{ ~~~~ Example: B. Bell Diagonal States (BDS):}}\\

\noindent A Bell diagonal state is a 2-qubit state that is diagonal in the Bell basis. In other words, it is a mixture of 
the four Bell states. It can be written as ${\rho}_{12} = p_1 |\Phi^+\rangle\langle\Phi^+| + p_2 |\Psi^+\rangle\langle\Psi^+| + p_3|\Phi^- \rangle\langle\Phi^- | + p_4|\Psi^- \rangle\langle\Psi^- | $, 
where $ p_1 + p_2 + p_3 +p_4 = 1 $. A Bell-diagonal state is separable if all the probabilities are less or equal to $1/2$. 
After transforming ${\rho}_{12}$ from Bell-basis to computational basis $\{|0\rangle, |1\rangle\}$, we convert the 4-parameter description to a 3-parameter description via $\beta$ coordinate transformation. 
In the Beta coordinate system, the set of Bell-diagonal states can be visualized as a tetrahedron where the four Bell 
states are the corners. The following change of coordinate system makes the plotting of states easy, which is given by : $\beta_0 = \frac{1}{2} (p_1 + p_2 + p_3 + p_4)$,  $\beta_1 = \frac{1}{2} (p_1 - p_2 - p_3 + p_4)$
$\beta_2 = \frac{1}{\sqrt{2}} (p_1 - p_4)$, $\beta_3 = \frac{1}{\sqrt{2}} (p_2 - p_3)$.\\

\noindent The coordinate $\beta_0$ will always be equal to $1/2$, and  ($\beta_1$, $\beta_2$ ,$\beta_3$) can be plotted 
in 3D easily. In $\beta$-coordinate system, Bell diagonal states are mathematically described as , $\rho_{12} = \big\{\vec{0},\vec{0},T^b\big\}$, where $\vec{0}$ is the null vector and $T^b=diag[ \sqrt{2} (\beta_2 - \beta_3), -2 \beta_1 , \sqrt{2} (\beta_2 + \beta_3)]$ is the correlation matrix .The $\beta$-coordinate system has the advantage that two of the edges are parallel to axes of the coordinate-system.\\

\textit{\centerline{Local Cloning}}\\

\noindent When we apply the general local unitary part  ($U_{l}^{1} \otimes U_{l}^{2}$) (see Appendix A) of the cloning machine, the local output states are given by, $\tilde\rho_{13} = \big\{\vec{0}, \vec{0}, T_l\big\}$, $\tilde\rho_{24} = \big\{\vec{0}, \vec{0}, T_l\big\}$ and the non-local output states are $\tilde\rho_{12} = \tilde\rho_{34} = \big\{\vec{0}, \vec{0}, \mu^2 T^b \big\}$, where $\mu=1-2\lambda$ with $\lambda$ being the machine parameter. $T_{l}$ is a 3x3 diagonal matrix with the diagonal elements being $2\lambda$, $2\lambda$ and $1 - 4\lambda$ ,whereas,  $T^{b}$ is the input state correlation matrix for Bell diagonal state, as described earlier.\\

\noindent In case of general local cloner, the coherence of local output states is given by, $C(\tilde\rho_{13}) = C(\tilde\rho_{24}) = 2\lambda$, while, on the other hand, the coherence of non-local output states is given by, $C(\tilde\rho_{12}) = C(\tilde\rho_{34}) = \frac{|(2 \beta_2 - \beta_3) (1 - 2 \lambda)^2|}{\sqrt{2}}$. Now, in order to obtain the state independent version of the local cloner, we substitute $\lambda=1/6$ (see Appendix A). Hence, after substituting $\lambda=1/6$, we have $C(\tilde\rho_{13}) = C(\tilde\rho_{24}) = \frac{1}{3}$ and $C(\tilde\rho_{12}) = C(\tilde\rho_{34}) = \frac{4 |2 \beta_2 - \beta_3| }{9\sqrt{2}}$. The  condition of broadcasting in this case is $ 4 |2 \beta_2 - \beta_3| > 3\sqrt{2}$.\\

\begin{table}[ht!]
\begin{center}
    \begin{tabular}{| p{0.8cm} | p{0.8cm} | p{4.5cm} |}
    \hline
$\beta_1$ & $\beta_3$ & $\beta_2$            \\ \hline
0.15 & -0.2 & 0.430 $<$ $\beta_2$ $\leq$ 0.460 \\ \hline
0.15 & -0.15 & 0.455 $<$ $\beta_2$ $\leq$ 0.460  \\ \hline
0.2 & 0.2 & -0.495 $\leq$ $\beta_2$ $<$ -0.430   \\ \hline
0.2 & 0.15 & -0.495 $\leq$ $\beta_2$ $<$ -0.455   \\ \hline
0.2 & 0.1 & -0.495 $\leq$ $\beta_2$ $<$ -0.480 \\ \hline
0.2 & -0.1 & 0.480 $<$ $\beta_2$ $\leq$ 0.495  \\ \hline
0.2 & -0.2 & 0.430 $<$ $\beta_2$ $\leq$ 0.495  \\ \hline
0.3 & 0.05 & -0.566 $\leq$ $\beta_2$ $<$ -0.505 and 0.555 $<$ $\beta_2$ $\leq$ 0.566  \\ \hline
0.3 & 0.1 & -0.566 $\leq$ $\beta_2$ $<$ -0.480 \\ \hline
0.3 & -0.1 & 0.48033 $<$ $\beta_2$ $\leq$ 0.566  \\ \hline
0.4 & 0.05 & -0.636 $\leq$ $\beta_2$ $<$ -0.505 and 0.555 $<$ $\beta_2$ $\leq$ 0.636  \\ \hline

\end{tabular}
\end{center}
\caption{ \noindent \scriptsize Broadcasting ranges for local state-independent cloning in case of Bell-diagonal in terms of $\beta_1$ , $\beta_3$ and $\beta_2$ parameters.}
\label{tab:BellSSDwenerp=0.80}
\end{table}

\noindent In Table \ref{tab:BellSSDwenerp=0.80}, we have listed the broadcasting range for local cloning in terms of $\beta_1$, $\beta_2$ and $\beta_3$.We fix the values of $\beta_1$ and $\beta_3$ and substitute their respective values in the expression : $4 |2 \beta_2 - \beta_3| > 3\sqrt{2} $ in order to arrive at the desired range of values in terms of $\beta_2$. However, for certain values of  $\beta_1$ and  $\beta_3$, we observe discontinuities in the intervals corresponding to range of  $\beta_2$, which suggest the presence of a no-broadcast zone  in between. 
\\

\begin{figure}[htbp]
\begin{center}
\includegraphics[height=7.5cm,width=8cm]{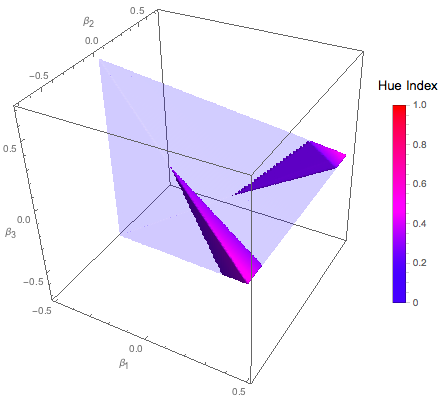}\\
\end{center}
\caption{\noindent \scriptsize Broadcasting of coherence in Bell-diagonal states via local state independent cloning.The Hue Index represents the trend in Hue with respect to value of non-local coherence ($C(\tilde\rho_{12})$ or $C(\tilde\rho_{34})$) after normalised in a scale of 0 to 1.}
\label{fig:Bell_SIDL}
\end{figure}

\noindent In Figure \ref{fig:Bell_SIDL}, we have superimposed the region of broadcasting on the tetrahedral-region hosting all Bell-diagonal states in $\beta$-coordinate system, with axes $\beta_1$, $\beta_2$ and $\beta_3$.The light blue region represents the no-broadcast zone, while the two regions on the dark blue - magenta spectrum, along the edges of the two corners of the tetrahedron, highlight the regions where broadcasting of coherence is possible. In these two regions along the edges, the change in hue signifies how the values of the non-local coherence, $C(\tilde\rho_{12})$ or $C(\tilde\rho_{34})$, changes with respect to the values of the $\beta$-parameters. As the value of $C(\tilde\rho_{12})$ or $C(\tilde\rho_{34})$ increases, the hue changes from blue to magenta. From the figure, it is evident that in terms of coverage of the tetrahedron's volume, the zone of broadcasting is very small when compared to the no-broadcast region, thereby implying that we can broadcast coherence only in a limited set of Bell-diagonal states via local state independent cloning.\\

\textit{\centerline{Non-local Cloning}}\\

\noindent Next, we use general non-local cloner 
(see Appendix A) for cloning Bell diagonal state. The output states in this case are given by, $\tilde\rho_{13} =  \big\{\vec{0},\vec{0},T_{nl}\big\}$, $\tilde\rho_{24} = \big\{\vec{0},\vec{0},T_{nl}\big\}$, and $\tilde\rho_{12} = \tilde\rho_{34} = \big\{\vec{0},\vec{0},\mu T^{b}\big\}$, where $\mu=1-4\lambda$ with $\lambda$ being the machine parameter. Here, $T_{nl}$ is a 3x3 diagonal matrix with the diagonal elements being $2\lambda$, $2\lambda$ and $1 -8\lambda$, whereas, $T^{b}$ is the input state correlation matrix for bell diagonal state, as described earlier.\\

\noindent The coherence of local output states is given by, $C(\tilde\rho_{13}) = C(\tilde\rho_{24}) = 2\lambda$, while, on the other hand, the coherence of non-local output states is given by, $C(\tilde\rho_{12}) = C(\tilde\rho_{34}) = \frac{1}{\sqrt{2}}[|(2 \beta_2 - \beta_3) (-1 + 4 \lambda)| + |\beta_3 - 4 \beta_3 \lambda|]$. To obtain the state independent version of the cloner, we substitute the value of $\lambda$ by $1/10$ (see Appendix A).  Now, after substituting $\lambda = 1/10$, we have, $C(\tilde\rho_{13}) = C(\tilde\rho_{24}) = \frac{1}{5}$, $C(\tilde\rho_{12}) = C(\tilde\rho_{34}) = \frac{3}{5\sqrt{2}}[|(\beta_3 - 2 \beta_2)| + |\beta_3|].$
In this case, the expression for the condition of broadcasting reduces to,  $3\sqrt{2}[|(\beta_3 - 2 \beta_2)| + |\beta_3|] > 1.$\\

\begin{table}[ht!]
\begin{center}
    \begin{tabular}{| p{0.8cm} | p{0.8cm} | p{4.5cm} |}
    \hline
$\beta_1$ & $\beta_3$ & $\beta_2$            \\ \hline

-0.2 & -0.2 & 0.036 $<$ $\beta_2$ $\leq$ 0.212 \\ \hline
-0.2 & -0.1 & 0.136 $<$ $\beta_2$ $\leq$ 0.212  \\ \hline
-0.2 & 0.1 & -0.212 $\leq$ $\beta_2$ $<$ -0.136  \\ \hline
-0.2 & 0.2 & -0.212 $\leq$ $\beta_2$ $<$ -0.036  \\ \hline

-0.1 & -0.2 & -0.283 $\leq$ $\beta_2$ $<$ -0.236 and 0.036 $<$ $\beta_2$ $\leq$ 0.283  \\ \hline
-0.1 & -0.1 & -0.283 $\leq$ $\beta_2$ $<$ -0.236 and 0.136 $<$ $\beta_2$ $\leq$ 0.283  \\ \hline
-0.1 & 0.1 & -0.283 $\leq$ $\beta_2$ $<$ -0.136 and 0.236 $<$ $\beta_2$ $\leq$ 0.283  \\ \hline
-0.1 & 0.2 & -0.283 $\leq$ $\beta_2$ $<$ -0.036 and 0.236 $<$ $\beta_2$ $\leq$ 0.283  \\ \hline

0.1 & -0.2 & -0.424 $\leq$ $\beta_2$ $<$ -0.236 and 0.036 $<$ $\beta_2$ $\leq$ 0.424  \\ \hline
0.1 & -0.1 & -0.424 $\leq$ $\beta_2$ $<$ -0.236 and 0.136 $<$ $\beta_2$ $\leq$ 0.424  \\ \hline
0.1 & 0.1 & -0.424 $\leq$ $\beta_2$ $<$ -0.136 and 0.236 $<$ $\beta_2$ $\leq$ 0.424 \\ \hline
0.1 &  0.2 & -0.424 $\leq$ $\beta_2$ $<$ -0.036 and 0.236 $<$ $\beta_2$ $\leq$ 0.424  \\ \hline
0.2 & -0.2 & -0.495 $\leq$ $\beta_2$ $<$ -0.236 or 0.036 $<$ $\beta_2$ $\leq$ 0.495  \\ \hline
0.2 & -0.1 & -0.495 $\leq$ $\beta_2$ $<$ -0.236 or 0.136 $<$ $\beta_2$ $\leq$ 0.495  \\ \hline
0.2 & 0.1 & -0.495 $\leq$ $\beta_2$ $<$ -0.136 or 0.236 $<$ $\beta_2$ $\leq$ 0.495  \\ \hline
0.2 & 0.2 & -0.495 $\leq$ $\beta_2$ $<$ -0.036 or 0.236 $<$ $\beta_2$ $\leq$ 0.495 \\ \hline

\end{tabular}
\end{center}
\caption{\noindent \scriptsize Broadcasting ranges for state-independent non-local cloning in case of Bell-diagonal states in terms of $\beta_1$ , $\beta_3$ and $\beta_2$ parameters.}
\label{tab:Bell_SIDNL_table}
\end{table}

\noindent  In Table \ref{tab:Bell_SIDNL_table}, we have shown some of the trends in broadcasting ranges obtained in terms of $\beta_1$, $\beta_2$ and $\beta_3$ after state independent non-local cloning transformations. We fix the values of $\beta_1$ and $\beta_3$ and substitute their respective values in the expression $3\sqrt{2}[|(\beta_3 - 2 \beta_2)| + |\beta_3|] > 1.$ Here, from $\beta_1 = -0.1$ onwards, discontinuities in the intervals corresponding to the ranges of $\beta_2$-parameter are observed, which imply the presence of a no-broadcast zone in between.\\

\begin{figure}[htbp]
\begin{center}
\includegraphics[height=7.5cm,width=8cm]{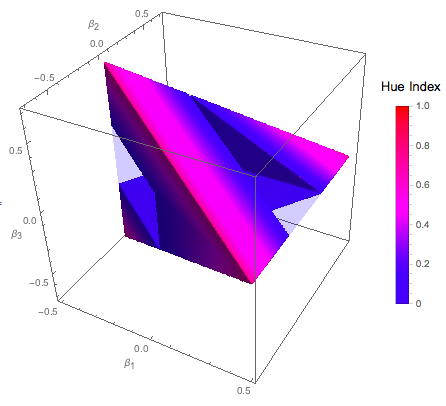}\\
\end{center}
\caption{\noindent \scriptsize Broadcasting ranges for non-local state-independent cloning in case of Bell-diagonal states.The Hue Index represents the trend in Hue with respect to value of non-local coherence ($C(\tilde\rho_{12})$ or $C(\tilde\rho_{34})$) after normalised in a scale of 0 to 1.}
\label{fig:Bell_SIDNL}
\end{figure}

\noindent In Figure \ref{fig:Bell_SIDNL}, the region of broadcasting is superimposed on the tetrahedral-region hosting all Bell-diagonal states in $\beta$-coordinate system, with axes $\beta_1$, $\beta_2$ and $\beta_3$.The light light regions represent the no-broadcast zones, while the regions on the dark blue - magenta spectrum, occupying a significant fraction of the tetrahedron's volume, highlight the regions where we can broadcast coherence. In the bluish-magenta region, the change in hue signifies how the values of the non-local coherence, $C(\tilde\rho_{12})$ or $C(\tilde\rho_{34})$, changes with respect to the values of the $\beta$-parameters. As the value of $C(\tilde\rho_{12})$ or $C(\tilde\rho_{34})$ increases, the hue changes from blue to magenta. For instance, the concentrated presence of magenta-pink hue, especially in the regions along the tetrahedron's edge and in the regions surrounding the corners, suggests that non-local coherence is much higher in those regions, compared to the other regions in the broadcasting zone. Also, we find that there is a light purple cavity-like zone sandwiched in between the bluish-magenta regions, which points to a no-broadcast zone. This cavity-like zone accounts for discontinuities int the intervals corresponding to ranges of $\beta_2$-parameter in Table \ref{tab:Bell_SIDNL_table}. Moreover, we observe in the figure that in terms of coverage of the tetrahedron's volume, the zone of broadcasting in the non-local cloning scenario is significantly larger when compared to that of the local cloning scenario (Figure \ref{fig:Bell_SIDL}), thereby suggesting that non-local state independent cloning performs much better when it comes to broadcasting of coherence in the bell-diagonal states.\\

\section{IV. Conclusion}

\noindent{In short, in this work, we have explored and studied the problem of broadcasting of coherence with the help of local and non-local cloning machines. We start with a incoherent pair of qubits along with a coherent pair of qubits as inputs to a quantum cloning machine and after cloning, we induce coherence in the initially incoherent pair of qubits. Firstly, we show that it is impossible to broadcast coherence optimally via cloning. However, we can broadcast coherence non-optimally via cloning. Secondly, we show that in case of non-optimal broadcasting, the coherence induced in the output pairs of qubits is always lesser than the initial coherent pair of qubits which was taken as input into the cloning machine. In this work, both these results have been proven with respect to the computational basis by taking the $l_{1}$-norm as quantifier of coherence. We conjecture that the impossibility of optimal broadcasting and the possibility of non-optimal broadcasting will hold true for any legitimate quantifier of coherence within the framework of resource theory of coherence. This is because for optimal broadcasting of coherence, by definition, we must have zero coherence in the local output states with respect to any legitimate quantifier of coherence in the chosen basis. Hence, in such a scenario, it is impossible for the local output states to be coherent with respect to one quantifier of coherence, while, being incoherent with another quantifier of coherence.\\}

\noindent \textit{Acknowledgements:} Authors acknowledge Mr. Chiranjib Mukhopadhya, Mr. Tushant Jha  and Mr. Sourav Chatterjee for their suggestions.


\section{Appendix: Supplementary Material}
\section{Appendix A}

\noindent The cloning of arbitrary quantum states is forbidden by No-cloning theorem. However, there still exists the possibility of cloning arbitrary quantum states approximately with certain fidelity, or, 'imperfection'. Quantum cloning transformations are (CPTP) completely positive trace preserving map between two quantum systems, supported by an ancilla state. In this section, at first, we briefly describe the Buzek-Hillery (B-H) QCM. \\

\noindent The unitary part of B-H cloning machine ($U_{bh}$) is a $M$-dimensional quantum copying transformation acting on a state $\left|\Psi_i\right\rangle_{a_0}$ ($i$ = 1, ..., $M$), which is to be copied onto a blank state $\left|0\right\rangle_{a_1}$. Initially, the copy is prepared in state $\left|X\right\rangle_x$ which subsequently gets transformed into another set of state vectors $\left|X_{ii}\right\rangle_x$ and $\left|Y_{ij}\right\rangle_x$ as a result of application of the cloning machine. Here $a_0$, $a_1$ and $x$ represent the input, blank and machine qubits respectively. In this case, these transformed state vectors belong to the orthonormal basis set in the $M$-dimensional space. The transformation scheme $U_{bh}$ is given by,
\begin{eqnarray}
&&U_{bh}\left|\Psi_i\right\rangle_{a_0} \left|0\right\rangle_{a_1} \left|X\right\rangle_x \rightarrow  c\left|\Psi_i\right\rangle_{a_0} \left|\Psi_i\right\rangle_{a_1} \left|X_{ii}\right\rangle_x \nonumber\\
&&+d\displaystyle \sum_{j\neq i}^{M} \left(\left|\Psi_i\right\rangle_{a_0} \left|\Psi_j\right\rangle_{a_1} +\left|\Psi_j\right\rangle_{a_0}\left|\Psi_i\right\rangle_{a_1}\right) \left|Y_{ij}\right\rangle_x,
\label{eq:B-H_gen_transform}
\end{eqnarray}
where $i,\:j$ = $\{1,...,M\}$,  and the coefficients $c$ and $d$ are real. 

\subsection{State independent cloning transformations}
\noindent An optimal state independent version of the B-H cloner can be obtained from Eq.\eqref{eq:B-H_gen_transform} by imposing the unitarity and normalization conditions which give rise to the following set of constraints, \\$\braket{X_{ii}}{X_{ii}} = \braket{Y_{ij}}{Y_{ij}} = \braket{X_{ii}}{Y_{ji}} = 1$,\\ $\braket{X_{ii}}{Y_{ij}} = \braket{Y_{ji}}{Y_{ij}} = \braket{X_{ii}}{X_{jj}} = 0$ and 
\begin{equation}
c^2 = \frac{2}{M+1}, d^2 = \frac{1}{2(M+1)}.
\label{eq:stateindependence}
\end{equation}
\noindent Here, $M=2^m$ where $m$ is the number of qubits in a given quantum register. In this transformation,the fidelity of the output copies doesn't depend on the input state because the shrinking or scaling property is independent of the input states.\\ 

\noindent In case of the state dependent cloner the unitarity constraints on the B-H cloning transformation in Eq.~\eqref{eq:B-H_gen_transform} give rise to the following set of conditions: 
\begin{eqnarray}
 & c^2\left\langle X_{ii}|X_{ii} \right\rangle + d^2\displaystyle \sum^{M}_{j\neq i} 2 \left\langle Y_{ij}|Y_{ij} \right\rangle =1,\:\left\langle Y_{ij}|Y_{kl} \right\rangle =0,
\end{eqnarray}
where $i \neq j$ and $ij \neq kl$ for  $i,j,k,l=\{1,...,M\}$; and the output states are orthonormal.\\
We find that,
\begin{eqnarray}
c^2 = 1 - 2(M-1)d^2 .
\end{eqnarray}
If we substitute $\lambda = d^2$, $c^2$ becomes $1-2(M-1)\lambda$. We call $\lambda$ as the machine parameter.

\subsubsection{Local state independent cloning machine}
\noindent The optimal cloning machine described above  with $M=2$ becomes a local cloning machine ($U^{l}_{bhsi}$). By substituting the values of the coefficients  $c$=$\sqrt{\frac{2}{3}}$ and $d$=$\sqrt{\frac{1}{6}}$ in Eq.~\eqref{eq:B-H_gen_transform}, we obtain the optimal state independent cloning machine which can be used for local copying purposes.In this case, the value of $\lambda$ is $1/6$.\\

\subsubsection{Non local state independent cloning machine}
\noindent Similarly,the optimal cloning machine described above becomes a non local cloning machine ($U^{nl}_{bhsi}$) for $M=4$. Then the corresponding 
values of the coefficients $c$ and $d$ in Eq.10 become $\sqrt{\frac{2}{5}}$ and $\sqrt{\frac{1}{10}}$ 
respectively.Now, by substituting the values of the coefficients $c$=$\sqrt{\frac{2}{5}}$ and $d$=$\sqrt{\frac{1}{10}}$ in Eq.~\eqref{eq:B-H_gen_transform}, we obtain the optimal state independent non local cloning machine.In this case, the value of $\lambda$ is $1/10$.

\section{Appendix B}
\label{B}
\begin{figure}[htbp]
\begin{center}
\[
\begin{array}{cc}
\includegraphics[height=4.5cm,width=4.5cm]{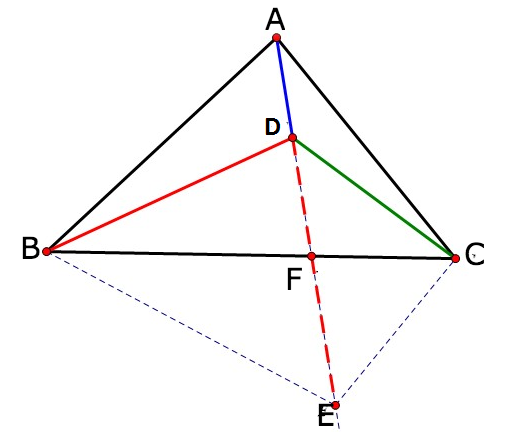} 
\end{array}
\]
\end{center}
\caption{\noindent \scriptsize
$\Delta ABC$ with point M in its interior.}
\end{figure}

\noindent \begin{lemma}
\textit{Let M be a point in the interior of $\Delta ABC.$ Then, $\overline{AC}+\overline{BC} > \overline{AD}+\overline{BD}$, $\overline{AB}+\overline{AC} > \overline{BD}+\overline{DC}$ and $\overline{BC}+\overline{BA} > \overline{DC}+\overline{DA}$.}
\end{lemma}
 
\noindent \begin{proof}
As shown in Figure 5, $\overline{AD}$ is extended to point E outside $\Delta ABC$ intersecting $\overline{BC}$ at $F$. By applying Triangle Inequality on $\Delta ACF$ and $\Delta BFD$, we get $\overline{AC}$ + $\overline{CF}$ $>$ $\overline{AF}$ and, $\overline{DF}$ + $\overline{BF}$ $>$ $\overline{BD}$ respectively.\\Also, as $D$ lies on $\overline{AF} $, we have ,$\overline{AD}$ + $\overline{DF}$ = $\overline{AF}.$ \\Now, adding these two inequalities, we have,
\begin{flushleft}
$\overline{AC}$ + $\overline{CF}$ +$\overline{BF}$ + $\overline{DF}$ $>$ $\overline{AF}$ + $\overline{BD}$ \\

$\implies$ $\overline{AC}$ + $(\overline{CF}$ +$\overline{BF})$ + $\overline{DF}$ $>$ $(\overline{AD}$ + $\overline{DF})$ + $\overline{BD}$ \\

$\implies$ $\overline{AC}$ + $\overline{BC}$ +  $\overline{DF}$ $>$ $\overline{AD}$ + $\overline{DF}$ + $\overline{BD}$ \\ 
\end{flushleft}
Cancelling out $\overline{DF}$, we finally get,$ \overline{AC}$ + $\overline{BC}$ $>$ $\overline{AD}$ + $\overline{BD}.$ Similarly, the other two inequalities can be proved.\\
\end{proof}


\begin{thebibliography}{99}

\bibitem{correlation1} Reinhard F. Werner,  \href{https://journals.aps.org/pra/abstract/10.1103/PhysRevA.40.4277}{Phys. Rev. A  \textbf{40, 4277 (1989)}}

\bibitem{correlation2} Zhang JS.\textit{et. al},  \href{http://www.naturalspublishing.com/Article.asp?ArtcID=1192}{Quantum Physics Letters \textbf{1, 69-77 (2012)}}



\bibitem{entanglement} A. Einstein \textit{et. al}, \href{https://journals.aps.org/pr/abstract/10.1103/PhysRev.47.777}{Phys. Rev. \textbf{47, 777 – (1935)}}

\bibitem{baumgratz} T. Baumgratz. \textit{et. al}, \href{http://dx.doi.org/10.1103/PhysRevLett.113.140401}{Phys. Rev. Lett. \textbf{113}, 140401 (2014)}

\bibitem{Girolami} Davide Girolami, \href{http://journals.aps.org/prl/abstract/10.1103/PhysRevLett.113.170401}{Phys. Rev. Lett. \textbf{113, 170401(2014)}}


\bibitem{Napoli}Carmine Napoli, \textit{et al}, \href{http://journals.aps.org/prl/abstract/10.1103/PhysRevLett.116.150502}{Phys. Rev. Lett. \textbf{116, 150502 (2016)}}


\bibitem{rob}C. Mukhopadhyay, U.K.Sharma and I. Chakrabarty, \href{https://arxiv.org/abs/1705.09080}{arXiv: \textbf{1705.09080}, 070402 (2017)}

  A note on robustness of coherence for multipartite quantum states, C. Mukhopadhyay, U.K.Sharma, I. Chakrabarty arXiv:1705.09080 (2017)


\bibitem{resource1} E. Chitambar, A. Streltsov, S. Rana, M. N. Bera, G. Adesso, and M. Lewenstein, \href{http://dx.doi.org/10.1103/PhysRevLett.116.070402}{Phys. Rev. Lett. \textbf{116}, 070402 (2016)}

\bibitem{resource2}A. Winter and D. Yang, \href{http://dx.doi.org/10.1103/PhysRevLett.116.120404}{Phys. Rev. Lett. \textbf{116}, 120404 (2016)}

\bibitem{resource3}I. Marvian and R.W. Spekkens, \href{http://dx.doi.org/10.1103/PhysRevA.90.062110}{Phys. Rev. A \textbf{90}, 062110,(2014)}

\bibitem{resource4}G. Gour, I. Marvian, and R.W. Spekkens,\href{http://dx.doi.org/10.1103/PhysRevA.80.012307}{Phys. Rev. A \textbf{80}, 012307 (2009)} 

\bibitem{resource5}S. Cheng and M.J.W. Hall, \href{http://dx.doi.org/10.1103/PhysRevA.92.042101}{Phys. Rev. A \textbf{92}, 042101 (2015)} 

\bibitem{resource6}M. N. Bera, \textit{et al} \href{http://dx.doi.org/10.1103/PhysRevA.92.012118}{Phys. Rev. A \textbf{92}, 012118 (2015)}

\bibitem{genuine}Alexander Streltsov, \textit{et al},  \href{http://iopscience.iop.org/article/10.1088/1751-8121/50/4/045301/meta}{J. Phys. A \textbf{50, 045301 (2017)}}

\bibitem{Yao}Yao Yao, \textit{et al}\href{http://journals.aps.org/pra/abstract/10.1103/PhysRevA.92.022112}{Phys. Rev. A \textbf{92, 022112(2015)}}

\bibitem{Chandrashekhar}R. Chandrashekar,
\textit{et al} \href{http://journals.aps.org/prl/abstract/10.1103/PhysRevLett.116.150504}{Phys. Rev. Lett. \textbf{116, 150504(2015)}}

\bibitem{wootters} W. K. Wootters and W. H. Zurek, \href{http://www.nature.com/nature/journal/v299/n5886/abs/299802a0.html}{Nature \textbf{299}, 802 (1982)}.

\bibitem{buzek1} V. Buzek and M. Hillery, \href{http://dx.doi.org/10.1103/PhysRevA.54.1844}{Phys. Rev. A \textbf{54}, 1844 (1996)}.

\bibitem{brub} D. Bruss, P.D. DiVincenzo, A. Ekert, C. A. Fuchs, C. Macchiavello and J. A. Smolin, \href{http://dx.doi.org/10.1103/PhysRevA.57.2368}{Phys. Rev. A \textbf{57}, 2368 (1998)}.

\bibitem{gisin2} N. Gisin, \href{doi:10.1016/S0375-9601(98)00170-4}{Phys. Lett. A \textbf{242}, 1 (1998)}. 

\bibitem{sourav} S. Chatterjee, S. Sazim, I. Chakrabarty,\href{http://dx.doi.org/10.1103/PhysRevA.93.042309}{Phys. Rev. A \textbf{93}, 042309 (2016)}.


\bibitem{aditya} A. Jain, I. Chakrabarty, S. Chatterjee
,\href{https://arxiv.org/abs/1702.02123}{arXiv : \textbf{1702.02123}, (2017)}.

\bibitem{Kher} A. Kheirollahi. et.al, \href{https://arxiv.org/abs/1604.04995}{arXiv: 1604.04995}.

\bibitem{barnum} H. Barnum, C. M. Caves, C. A. Fuchs, R. Jozsa and B. Schumacher, \href{http://dx.doi.org/10.1103/PhysRevLett.76.2818}{Phys. Rev. Lett. \textbf{76}, 2818 (1996)}. 

\bibitem{piani} M. Piani, P. Horodecki and R. Horodecki, \href{http://dx.doi.org/10.1103/PhysRevLett.100.090502}{Phys. Rev. Lett. \textbf{100}, 090502 (2008)}. 
\bibitem{barnum-gen} H. Barnum, J. Barrett, M. Leifer and A. Wilce, \href{http://dx.doi.org/10.1103/PhysRevLett.99.240501}{Phys. Rev. Lett. \textbf{99}, 240501 (2007)}. 

\bibitem{luo} S. Luo, \href{http://link.springer.com/article/10.1007/s11005-010-0389-1}{Lett. Math. Phys. \textbf{92}, 143 (2010)}. 

\bibitem{luo-li} S. Luo, N. Li and X. Cao, \href{http://dx.doi.org/10.1103/PhysRevA.79.054305}{Phys. Rev. A \textbf{79}, 054305 (2009)}.

\bibitem{Uttam}Alexander Streltsov, \textit{et al},  \href{http://journals.aps.org/prl/abstract/10.1103/PhysRevLett.115.020403}{Phys. Rev. Lett. \textbf{115, 020403}}

\bibitem{Ma}Jiajun Ma, \textit{et al},  \href{http://journals.aps.org/prl/abstract/10.1103/PhysRevLett.116.160407}{Phys. Rev. Lett. \textbf{116, 160407}}



\end{thebibliography}
\end{document}